\PassOptionsToPackage{dvipsnames}{xcolor}
\documentclass[sigplan,10pt]{acmart}
\usepackage{marginnote, fontawesome5, dsfont}
\usepackage{tikz}
\usetikzlibrary{positioning, arrows.meta, fit}
\usepackage{simplenotes}
\usepackage{cleveref}

\usepackage{xspace}
\usepackage{enumitem}
\setlist[itemize]{leftmargin=*, itemsep=0.25em, topsep=0.25em}

\usepackage{newunicodechar}
\newunicodechar{π}{\ensuremath{\pi}}
\newunicodechar{₁}{\ensuremath{{}_1}}
\newunicodechar{₂}{\ensuremath{{}_2}}

\usepackage{lstlean}

\makeatletter
\AtBeginEnvironment{leanlisting}{\upshape\let\@vspace\@vspace@orig\let\@vspacer\@vspacer@orig}
\makeatother

\definecolor{commentcolor}{named}{Gray}
\definecolor{bodycolor}{named}{Black}
\definecolor{identifiercolor}{named}{Blue}
\colorlet{keywordcolor}{Orange!85!Red}
\colorlet{sortcolor}{Orange!85!Red}
\colorlet{tacticcolor}{Orange!85!Red}
\definecolor{symbolcolor}{named}{Black}

\providecommand{\bigwedgeop}{\bigwedge}
\providecommand{\bigveeop}{\bigvee}
\providecommand{\notni}{\not\ni}

\lstnewenvironment{leanlisting}[1][]{\lstset{language=Lean,#1}}{}

\newcommand{\leaninl}{\lstinline[language=lean]}

\newcommand\eg{\textit{e.g.\@},\xspace}

\newcommand\zflean{\textsc{ZFLean}\xspace}

\newcommand{\setcomp}[2]{\ensuremath{\{\, #1 \mid #2 \,\}}}
\newcommand{\sem}[1]{\ensuremath{[\![#1]\!]}}
\newcommand{\csem}[1]{\ensuremath{[\![#1]\!]}}

\newcommand{\irrelprf}{\ensuremath{\cdot\mkern-3.3mu\cdot\mkern-3.3mu\cdot}}

\def\Dom{\mathsf{Dom}}
\def\Ran{\mathsf{Range}}
\newcommand{\Image}[2]{#1[#2]}
\def\IsFunc{\mathsf{IsFunc}}
\def\IsPFunc{\mathsf{IsPFunc}}
\def\IsInjective{\mathsf{IsInjective}}
\def\IsSurjective{\mathsf{IsSurjective}}
\def\IsBijective{\mathsf{IsBijective}}

\newcommand{\Id}[1]{\mathds{1}_{#1}}
\newcommand{\fapplysym}{\ensuremath{\raisebox{-0.25ex}{@}^\mathsf{z}}}
\newcommand{\fapply}[2]{\ensuremath{\fapplysym\mkern-1.5mu#1}(#2)}

\newcommand{\zbinop}[1]{\mathbin{\ensuremath{\mathsf{#1}^{\mathsf{z}}}}}
\newcommand{\fcompZ}{\zbinop{\circ}}  
\newcommand{\lamZ}{\mathop{\ensuremath{\lambda^{\mathsf{z}}}}}
\newcommand\zlambda[4]{%
  \ensuremath{%
    \setlength{\arraycolsep}{1pt}%
    \begin{array}{ccccc}
      \lamZ & :    & #1 & \to     & #2 \\
                    & & #3 & \mapsto & #4
    \end{array}%
  }}

\renewcommand\implies{\Rightarrow}

\newcommand{\textlean}{%
  \raisebox{-2.5pt}{%
    \mbox{%
      \kern.2em\relax
      \includegraphics[height=1.6\fontcharht\font`X]{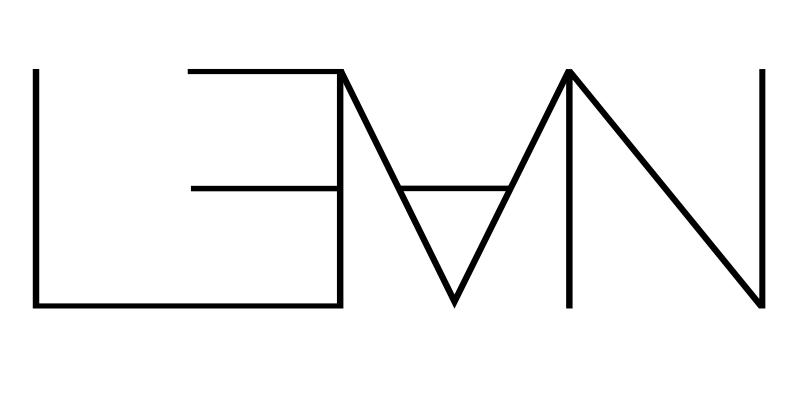}%
      \kern.2em\relax
    }%
  }%
}

\makeatletter
\AtBeginDocument{%
  \theoremstyle{acmdefinition}%
  \newtheorem{remark}[theorem]{Remark}%
}
\makeatother

\acmConference[]{}{}{}
\copyrightyear{2027}
\setcopyright{none}

\begin{document}

\title{ZFLean: A Framework for Set-Level Mathematics in Lean}

\author{Vincent Trélat}
\orcid{0009-0006-4143-3939}
\affiliation{%
  \institution{Université de Lorraine, CNRS, Inria, LORIA}
  \city{Nancy}
  \country{France}}
\email{vincent.trelat@inria.fr} 

\author{Matteo Pouillat}
\affiliation{%
  \institution{Télécom Nancy, Université de Lorraine}
  \city{Nancy}
  \country{France}}

\begin{abstract}
We present \zflean, a Lean~4 library for doing core mathematics inside Mathlib's model of ZFC set theory.
On top of this model, we build a relational calculus for binary relations, partial functions, and total functions, and the canonical set-theoretic constructions (Booleans, naturals, integers, rationals, sums, and options) with their usual algebraic structures.
Four small, bounded tactics discharge the side conditions these definitions generate: being a relation, a partial function, or a total function, and belonging to a function's domain, the set-theoretic form of definedness.
A thin typed layer of refinement and quotient types over the set-theoretic universe lets Lean's rewriting, type classes, and algebra tactics apply to the typed views of set-level objects.
Transfers to Lean's native types restate a goal about these objects as the equivalent goal about their Mathlib counterpart, where Mathlib's lemmas and decision procedures apply.
A case study written as a client of the library, a denotational semantics of a small imperative language on set-level stores, exercises the tactics and the transfers: definedness is a theorem about domains, determinism a partial-functionality proof, and program equivalence an equality of sets.
The library is provided as anonymous supplementary material.
\end{abstract}



\maketitle

\section{Introduction}
\label{sec:zflean-intro}

Modern mathematical developments are increasingly carried out with proof assistants, whose foundational cores range from first-order set theories (as in Mizar~\cite{mizar-sota,mizar-overview} or Isabelle/ZF~\cite{paulson-zf}) to more expressive systems based on simple type theory (Isabelle/HOL) or dependent type theory (Rocq~\cite{RocqReferenceManual}, Agda~\cite{AgdaReferenceManual}, and Lean~\cite{lean}).
The \emph{proofs-as-programs} vision of the Curry-Howard correspondence drove this evolution, which settled into two broad and often contrasted families of foundations: \emph{collection}-based theories (set-theoretic) and \emph{construction}-based theories (type-theoretic).
Various proposals aim to reconcile these viewpoints, either by generalizing over them, as category theory does, or by refining them, as constructive set theories such as CZF~\cite{aczel-czf} do.

Formalizing a logic or a programming language usually means giving it a semantics, and a denotational semantics in the tradition of Scott and Strachey~\cite{scott-strachey} is set-theoretic: syntax is interpreted into sets, and the interpretation itself is a partial map between them.
Set-based specification methods such as B~\cite{AbrialBBook}, Event-B~\cite{AbrialEventB}, Z~\cite{spivey-z}, and TLA\textsuperscript{+}~\cite{lamport-tla} state their proof obligations in the same language, untyped sets, relations, and partial functions with explicit definedness conditions, so a host for set-level reasoning is a natural target for their mechanized semantics.
Tr\'elat's translator BEer~\cite{beer} mechanizes one such semantics in \zflean, giving B proof obligations and their SMT-LIB encodings their denotations in a single ZFC universe, so that correctness is a set-theoretic statement about the two denotations.
Bergeron et~al.'s verified compiler for Distributed PlusCal~\cite{fugue}, formalized in Lean~4, is likewise adopting \zflean for its set-theoretic semantics.
Set-level formalization also serves mathematics itself: extensional equalities, partiality via relations, and category-style ``up to isomorphism'' reasoning.
However, working directly at set level in a typed assistant can be verbose and brittle: partial maps require encodings; extensionality must be disciplined; and crossing the boundary to native types is manual and error-prone.
Mathlib's own model of ZFC has served mostly for metatheory (Han and van Doorn's forcing construction and their independence proof for the continuum hypothesis~\cite{han-vandoorn-forcing,han-vandoorn-independence}) rather than as an environment for developing set-level mathematics.
We therefore build on Lean's Mathlib~\cite{mathlib} and the model of ZFC it provides, and develop \zflean, a framework for doing mathematics in Mathlib's model of ZFC with the ergonomics Lean users expect from typed Mathlib developments: notations, rewriting, automation, tactics, and typeclass-driven algebra.
The aim is \emph{not} to replace Mathlib's native arithmetic or data types.
\zflean contributes the following:
\begin{itemize}
  \item \emph{a relational calculus}: converse, identity, composition, domain, range, and images, together with their interaction laws and rewriting hints feeding four small tactics that automatically discharge well-formedness and domain-membership side conditions;
  \item \emph{canonical constructions and usual properties}: $\mathbb{B}$, $\mathbb{N}$, $\mathbb{Z}$, $\mathbb{Q}$, functions, sums, and options;
  \item \emph{transfers across the set/type boundary}: equivalences that carry the algebraic structure of these constructions to their Lean counterparts, and a \leaninl|transfer| tactic that restates a goal about these constructions as its equivalent goal about their Lean counterpart, where Mathlib's lemmas and decision procedures apply.
\end{itemize}
The transfers are part of the design, but \zflean is intended first as a set-level proving environment: every canonical construction has an explicit set-level anchor, refinement and quotient types provide typed views of these anchors with the algebraic structure exposed on the views (\Cref{tab:layers}), and interoperability is a convenience rather than a dependency for the set-level core.
The individual ingredients are standard: the constructions are textbook, the search engine behind the tactics is \leaninl|solve_by_elim|, and transport along equivalences is well understood. The contribution is the integration architecture that makes them one usable layer, a common untyped universe, typed views for the algebra, and equivalence-based transfer, tied together by the certificate-passing partial-function interface.
Throughout the paper, a ``\(\mathsf{z}\)'' superscript distinguishes ZFC-level operations from their Lean-level counterparts (\eg $\fcompZ$, $\fapplysym$, $\lamZ$).
The library is the artifact of this submission: it is provided as anonymous supplementary material, with build instructions, and is browsable at \url{https://anonymous.4open.science/r/ZFLean}.

We first recall Mathlib's ZFC model in \Cref{sec:zflean-background} and state what \zflean adds to it.
\Cref{sec:zflean-rel} presents the relational calculus together with its automation and tactics.
\Cref{sec:zflean-canon} surveys the canonical constructions \zflean builds on this calculus, together with their usual properties.
\Cref{subsec:iso-bridges} formalizes embeddings and isomorphisms, including the Cantor--Schr\"oder--Bernstein theorem, and presents the transfers to Lean types with the \leaninl|transfer| tactic.
\Cref{sec:zflean-case-study} evaluates the library on a case study written as a client of it, a denotational semantics on set-level stores, and reports which side conditions the tactics discharge, the size and elaboration time of the artifact, and its limitations.
We discuss related work in \Cref{sec:zflean-related}.

\section{Background and Design Choices}
\label{sec:zflean-background}

We rely on Mathlib's implementation of a model of ZFC set theory~\cite{mathlib-zfc}, which we briefly recall here.
The construction starts from a universe-polymorphic type of \emph{pre-sets} \leaninl|PSet| and its extensional quotient, to obtain the intended set-theoretic model.
Its central definitions are extensional equivalence \(X \sim Y\) and membership \(x \in X\).

\subsection{Pre-sets}

Pre-sets are defined inductively as universe-polymorphic structures.
\begin{definition}[Pre-set]
  A pre-set is built from a type \(\alpha\) in a universe of level \(u\)\footnote{\leaninl{Type u} denotes Lean's non-cumulative universe at level \leaninl{u}; the parameter \leaninl{u} makes the construction universe-polymorphic.} and a family \leaninl{A : α → PSet} interpreting the elements of \(\alpha\) as pre-sets:
\begin{leanlisting}
inductive PSet : Type (u + 1)
| mk (α : Type u) (A : α → PSet) : PSet
\end{leanlisting}
\end{definition}%
The empty pre-set is built on the empty type and its elimination principle: \leaninl{PSet.mk $\bot$ $\mathsf{elim}_\bot$}; every other pre-set is built from previously built ones.
These structures are called \emph{pre}-sets, rather than sets, because they do not yet carry an extensional notion of identity: the same collection may be described by many distinct pre-sets.
Two pre-sets \(X \eqcolon \langle \alpha, A \rangle\) and \(Y \eqcolon \langle \beta, B \rangle\) are \emph{extensionally equivalent}, written \(X \sim Y\), when every element of \(A\) is (inductively) equivalent to some element of \(B\) and vice versa ({\leaninl{PSet.Equiv}{}}); \(\sim\) is an equivalence relation, and it witnesses the axiom of extensionality.
Membership is defined from it, as \({x \in X \triangleq \exists\, a,\, x \sim A(a)}\), and is shown to be well-founded (there is no infinite descending \(\in\)-chain), which witnesses the axiom of regularity. Cartesian product, union, and powerset are defined on pre-sets as well, all to be lifted to sets.

\subsection{Sets}
\label{subsec:zflean-sets}

Lean has built-in quotient types, so sets are defined as the extensional quotient of pre-sets.
\begin{definition}[Set]
  The type of sets is defined as the quotient type \(\mathsf{ZFSet} \triangleq \mathsf{PSet}/{\sim}\).
\end{definition}

Operations on sets are then defined by lifting the corresponding operations on pre-sets via the quotient map.
Separation is one of them.
It is written \leaninl|ZFSet.sep p x| and characterized by \leaninl|mem_sep : y ∈ ZFSet.sep p x ↔ y ∈ x ∧ p y|, so it denotes exactly the comprehension \(\setcomp{y \in x}{p(y)}\), which is how we write it below.

The quotient definition already marks the divergence between \emph{intensional} and \emph{extensional} theories: ZFC sets are constrained by a notion of equality that is coarser than \emph{definitional} equality, which is central to systems such as Lean, Rocq, and Agda.
Although most proof assistants are based on intensional type systems, some like F\(^\star\)~\cite{fstar} treat provably equal terms like definitionally equal ones.
Our development targets the friction this divergence causes: the quotient-lifting boilerplate is written once, so proof scripts stay at the \leaninl|ZFSet| interface and reason \emph{extensionally} inside Lean while its core stays \emph{intensional}.

\begin{remark}
  \label[remark]{rmk:choice}
  The construction validates the full ZFC axioms, including choice.
  Lean's \leaninl|Classical.choice| is an ambient choice principle over all Lean types, external to the model: applied to it, it yields a witness \emph{term} usable in definitions, where ZFC's axiom of choice provides choice sets and functions inside the theory.
  \zflean invokes it in definitions in three disciplined ways.
  Most uses are \emph{definite descriptions}: noncomputable selectors pick an element that a functionality proof makes unique, as function evaluation on partial functions does (\Cref{subsec:zflean-func-eval-lambda}).
  Others are \emph{totalized primitives}, low-level exported operations like the projections below: the selector of the \leaninl{Option} encoding returns an arbitrary member on \leaninl{none}.
  The cardinal of a finite set selects a bijection with an initial segment of $\mathbb{N}$, and the selection is discharged by a specification and a uniqueness theorem: the set is in bijection with its cardinal, and any natural in bijection with it \emph{is} its cardinal (\leaninl|Card.spec|, \leaninl|Card.uniq|).
  Two selections are arbitrary by nature and guarded instead: the representative of a class of the set-level quotient of \Cref{sec:zflean-canon} is private, and the public lifting operation demands a proof of representative-independence; and the inductive set that seeds the construction of $\mathbb{N}$ is arbitrary, with the arbitrariness discharged by an invariance property (\Cref{rmk:nat-agnostic}).
  The $\varepsilon$ operator and the totalized primitives are exported but low-level; the two guarded selections above are the only arbitrary choices the interface itself consumes, and the results depend on \leaninl|Classical.choice| through these definitions and through Mathlib's own use of it.
\end{remark}

The model also provides Kuratowski's ordered pairs~\cite{kuratowski-pair}, \((x, y) \triangleq \bigl\lbrace\lbrace x\rbrace, \lbrace x, y\rbrace\bigr\rbrace\), but no projections or simplification lemmas; the projections \(\pi_1\), \(\pi_2\) and their expected properties are the first layer \zflean adds on top of the model.

\begin{definition}[Projections]
  \label[definition]{def:zflean-projections}
  We define general projections \(\pi_1\) and \(\pi_2\) for any set \(w\):
  \begin{equation*}
    \pi_1\, w \;\triangleq\; \bigcup {\bigcap w} \qquad \pi_2\, w \;\triangleq\;
    \begin{cases}
      \pi_1\, w      & \text{if}\ D_w = \varnothing \\
      \bigcup D_w    & \text{otherwise}
    \end{cases}
  \end{equation*}
  where \(D_w \triangleq \left(\bigcup w\right) \setminus \bigcap w\).
\end{definition}

\begin{remark}
  The projections \(\pi_1\) and \(\pi_2\) are \emph{total}: they are defined on every set, so no partiality enters the object language.
  On a set \(w\) that is not a pair they return some set with no intended meaning.
  This is harmless because every lemma about them assumes that \(w\) is a pair, either directly, \(w = (a, b)\), or through membership in a product, \(w \in A \times B\); their values outside pairs are never relied upon.
\end{remark}

With these projections, we prove the expected simplification lemmas for ordered pairs, shown below with proofs omitted.
\begin{leanlisting}
@[simp] theorem π₁_pair (x y : ZFSet) : π₁ (x.pair y) = x
@[simp] theorem π₂_pair (x y : ZFSet) : π₂ (x.pair y) = y
theorem pair_eta {z A B : ZFSet} (h : z ∈ A.prod B) :
  z = z.π₁.pair z.π₂
\end{leanlisting}

ZFC set theory has no built-in notion of function, so one must be defined.
In programming languages and type theories, functions are usually first-class and \emph{total} by construction, while mathematical functions can be \emph{partial} and are defined as \emph{functional relations} whose domain may be smaller than the intended source set.
This difference is usually bridged by using options, dependent types, or default ``junk'' values to encode partiality. In this work we aim to stay as close as possible to the set-theoretic definitions.
Mathlib already provides a definition \leaninl{IsFunc} for total functions and exponential objects, which we reuse in our development.
\begin{definition}[Total function]
  A set \(f\) is a total function from a set \(A\) to a set \(B\), written \(\IsFunc(A, B, f)\) (\leaninl|IsFunc| in Lean), if it is a functional relation whose domain is exactly \(A\); the exponential object \(B^A\) (\leaninl|funs A B|) separates \(\mathcal{P}(A\times B)\) by that predicate:
  \begin{align*}
    \IsFunc(A, B, f) & \;\triangleq\; f \subseteq A \times B \land \forall\, x \in A,\, \exists!\, y,\, (x, y) \in f \\
    B^A              & \;\triangleq\; \setcomp{f \in \mathcal{P}(A \times B)}{\IsFunc(A, B, f)}
  \end{align*}
\end{definition}

This definition does not cover partial functions, which are ubiquitous in mathematical developments; this motivates the relational calculus presented in the next section.

The model in Mathlib provides pre-sets, sets, the lifted basic operations (membership, separation, union, powerset, cartesian product), Kuratowski pairs, the ordinal \(\omega\), and the definition of total functions with the exponential \(B^A\) recalled above.
Everything described from here on is added by \zflean: the projections above, the relational calculus and its tactics, the canonical constructions, and the transfers to native types.

Every result of the set-level core is a Lean theorem about one model of ZFC, a semantic statement about Mathlib's \leaninl|ZFSet| established in Lean's type theory with its three standard axioms, \leaninl|propext|, \leaninl|Classical.choice|, and \leaninl|Quot.sound|, on which the results named in this paper depend and on nothing else.
They are not derivations in a formalized first-order ZFC calculus: the model validates separation for every Lean predicate on \leaninl|ZFSet| and replacement for every Lean function on it, through a classical instance that makes every such function definable, so the interface admits instances of these schemes that no first-order formula expresses.
A syntactic proof calculus for ZFC, as in Metamath or Lisa (\Cref{sec:zflean-related}), is a different object, and we claim nothing about derivability in it.

Three design conventions govern the development and are referred to throughout.
First, objects are formed in the set-theoretic universe, and typed views are interfaces rather than replacements; \Cref{sec:zflean-canon} states, per object, the set-level anchor.
Second, no junk value crosses the calculus interface: totalized primitives such as the projections above sit below it, and every statement about them carries a well-formedness hypothesis, which the tactics of \Cref{subsec:zflean-automation} discharge at the point of use.
Third, choice is confined by convention: a definition of the interface may use it only as a definite description, in a totalized primitive below the interface, or together with a proof that its result does not depend on the witness (\Cref{rmk:choice}).

\section{A Relational Calculus in ZFC}
\label{sec:zflean-rel}

Throughout this section, we extend the ZFC model provided by Mathlib with a calculus for binary relations and partial/total functions.
A \emph{binary relation} $R$ between $A$ and $B$ is a subset $R \subseteq A\times B$.
Our calculus provides the standard operations on relations (converse, identity, composition, domain, range, image) together with predicates for partial and total functions.
Automation tactics discharge the well-formedness side conditions these operations carry, so that the objects can be manipulated directly, much as in a paper proof.
These tactics are named \leaninl|zrel|, \leaninl|zpfun|, \leaninl|zfun|, and \leaninl|zdom|, and \Cref{subsec:zflean-automation} describes them.

\subsection{Basic definitions}

\begin{definition}[Converse, identity, composition]
  \label[definition]{def:rel-basic}
  The three base operations are defined as follows.
  \begin{itemize}
    \item The \emph{identity relation} $\Id{A}$ on a set $A$ collects the diagonal pairs of $A \times A$:
          \begin{equation*}
            \Id{A} \;\triangleq\; \setcomp{x \in A \times A}{\pi_1(x) = \pi_2(x)}
          \end{equation*}
          It is shown to be a total function belonging to $A^A$, and later a bijection and the neutral element for composition.
    \item Given a relation $R\subseteq A\times B$, its \emph{converse} $R^{-1}$ swaps the components of every pair:
          \begin{equation*}
            R^{-1} \;\triangleq\; \setcomp{z \in B\times A}{(\pi_2(z), \pi_1(z))\in R}
          \end{equation*}
          The side condition $R\subseteq A\times B$ ties the implicit sets $A$ and $B$, needed to form the ambient product $B \times A$, to $R$; in \zflean it is an automatic parameter whose default value \leaninl|by zrel| tries to discharge it at each use (see \Cref{subsec:zflean-automation}).
    \item Given relations $R\subseteq A\times B$ and $S\subseteq B\times C$, their \emph{composition} $S\circ R \subseteq A\times C$ ({\leaninl{composition}{}}) pairs $x$ with $z$ whenever some $y \in B$ links them:
          \begin{multline*}
            S \circ R \;\triangleq\; \{\, w \in A\times C \mid \exists\, y \in B,\, \\ (\pi_1(w), y)\in R \land (y, \pi_2(w))\in S \,\}
          \end{multline*}
          When $R$ and $S$ are arbitrary relations, it computes their expected composition; any pairs of $R$ or $S$ lying outside $A\times B$ or $B\times C$ are ignored.
          For total functions we add a dedicated notation $\fcompZ$ whose arguments carry the two functionality proofs, discharged by \leaninl|zfun|.
          The notation carries the proof that $S\fcompZ R$ is again a function, so a composite can be applied, rewritten, and composed again without the caller reproving it.
  \end{itemize}
\end{definition}

\begin{definition}[Domain, range, image]
  \label[definition]{def:dom-range-image}
          For a relation $R\subseteq A\times B$, its \emph{domain}, \emph{range}, and \emph{image} of a set $X\subseteq A$ are defined as:
  \begin{align*}
      \Dom(R)      & \;\triangleq\;  \setcomp{x\in A}{\exists\, y\in B,\, (x,y)\in R}          \\
      \Ran(R)      & \;\triangleq\;  \setcomp{y \in B}{\exists\, x \in \Dom(R),\, (x,y) \in R} \\
      \Image{R}{X} & \;\triangleq\;  \setcomp{y \in B}{\exists\, x \in X,\, (x, y) \in R}
  \end{align*}
  In all three, the side condition that $R$ is a relation between $A$ and $B$ is an automatic parameter discharged by \leaninl|zrel|.
\end{definition}

Partial functions are functional relations without a totality requirement.
\begin{definition}[Partial functions]
  A set $f$ is said to be a partial function from set $A$ to set $B$, written $\IsPFunc(f,A,B)$,\footnote{The order of arguments is deliberately changed compared to $\IsFunc$, so that in Lean, one may write \leaninl|f.IsPFunc A B|.} if it is \emph{functional}:
  \begin{multline*}
    \IsPFunc(f, A, B) \;\triangleq\; f \subseteq A \times B \,\land\, \bigl(\forall\, x,\,\\ \forall\, y,\, \forall\, z,\, (x,y) \in f \land (x,z) \in f \implies y = z\bigr)
  \end{multline*}
\end{definition}

We also derive the usual injectivity, surjectivity, and bijectivity predicates ($\IsInjective$, $\IsSurjective$, $\IsBijective$), which all invoke the tactic \leaninl|zfun| to discharge functional totality requirements.

\begin{remark}
  These predicates are therefore stated for \emph{total} functions. Although they could be stated for partial functions, we enforce totality by design. The motivation is practical: for a partial function, injectivity and surjectivity would be stated on its domain, on which it is total.
\end{remark}

\subsection{Function evaluation and λ-abstraction}
\label{subsec:zflean-func-eval-lambda}

We now add two notions that make functions practical to use: function application and $\lambda$-abstraction.
In a purely set-theoretic setting these are not primitive (a function \emph{is} its graph, and application recovers the unique image of a point), yet they are indispensable in practice; we define them to keep standard mathematical notation while automating the side conditions:
\begin{itemize}
  \item \emph{function application}, written \leaninl|@ᶻf ⟨x, h⟩|, maps \leaninl{x} to the unique \leaninl{y} such that \leaninl{(x, y) ∈ f}.
  Here \leaninl{@ᶻ} is a noncomputable dependently typed operator, not a function applied to a pair: it takes the point \leaninl{x} together with an explicit proof \leaninl{h : x ∈ Dom f}, and invokes the \leaninl|zpfun| tactic to prove that \leaninl{f} is a partial function.
  In the artifact this often occurs as \leaninl|@ᶻf ⟨x, by zdom⟩|, with the \leaninl|zdom| tactic of \Cref{subsec:zflean-automation} discharging the proof of domain membership.
  The definition selects the element \leaninl{y} with \leaninl{(x, y) ∈ f} through the ambient \leaninl|Classical.choice|: there is at most one such \leaninl{y} because \leaninl{f} is functional, and at least one because \leaninl{x} belongs to its domain (\leaninl{exists_unique_of_mem_dom}), so this is \emph{unique choice}, a definite description implemented with the ambient operator rather than an appeal to choice inside the object theory (\Cref{rmk:choice}).
  An equivalent alternative would be the set \(\bigcup \setcomp{y \in B}{(x,y) \in f}\), which is also noncomputable; we prefer the more direct formulation.

  Evaluating a composite reduces to composing the evaluations (\leaninl{fapply_composition}), and an application coincides with the unique element of the image of a singleton (\leaninl{fapply_eq_Image_singleton}):
  \begin{align*}
    \fapply{(g \fcompZ f)}{\langle x, \irrelprf \rangle} & = \fapply{g}{\langle \fapply{f}{\langle x, \irrelprf \rangle}, \irrelprf \rangle} \\
    \fapply{f}{\langle x, \irrelprf \rangle}             & = \textstyle\bigcap f[\{x\}]
  \end{align*}

  \item $\lambda$-abstraction, whose binder we write $\lamZ$, constructs functions from expressions: the abstraction
    \begin{equation*}
      \zlambda{A}{B}{x}{e(x)}
    \end{equation*}
    denotes the relation $\setcomp{(x,y)\in A\times B}{y=e(x)}$,
    which is then shown to belong to $B^A$ whenever $e(x)\in B$ for all $x\in A$.
    The binder may name the proof \leaninl|h : x ∈ A| explicitly by writing \leaninl|h : x ↦ e(x)|, or leave it unnamed. In both cases, the proof is brought into scope in the body $e$, so that it can be picked up by the \leaninl{zdom} tactic.
\end{itemize}

A common alternative to represent partial functions is to make every function \emph{total} by returning ``junk values'', fixed default values, outside the intended domain; Mathlib sometimes takes this approach, and usually marks it by appending ``\leaninl{D}'' (for default) to the identifier's name.
We deliberately avoid this approach for function application: a junk value would let $f$ be applied outside its domain, where the function is mathematically undefined, and give that application a value, a device for reconciling partial mathematical functions with the totality dependent type theory requires rather than a faithful set-theoretic notion.
Staying close to the set-theoretic principles, we instead adopt a \emph{by-construction} discipline: objects are formed together with proofs that they are meaningful.

\begin{remark}
  \label[remark]{rmk:no-junk}
  The price of this choice is the proof obligations that decorate our operators; the automation layer of \Cref{subsec:zflean-automation} keeps that price hidden in practice, discharging these obligations structurally so that the user writes the mathematics rather than the bookkeeping. Given a total function \leaninl|f| from \leaninl|A| to \leaninl|B|, one writes:
\begin{leanlisting}
λᶻ : A → B
   | x ↦ @ᶻf ⟨x, by zdom⟩
\end{leanlisting}
and the fact that function application expects a proof that \leaninl{f} is a partial function, and that \leaninl{x} belongs to its domain, is discharged automatically by the tactics \leaninl|zpfun| and \leaninl|zdom|.
\end{remark}

We prove the expected equational properties for this notation, summarized in the following theorem.

\begin{theorem}[Basic properties]
  \label{thm:lambda}
  Let $A$, $B$ be sets, $f \in B^A$, and $e, e' \colon \mathsf{ZFSet} \to \mathsf{ZFSet}$ with $e(x) \in B$ and $e'(x) \in B$ for every $x \in A$.
  \begin{itemize}
    \item \textbf{Specification:} For any sets $x$ and $y$:
          \begin{multline*}
            (x, y) \in \left(\zlambda{A}{B}{x}{e(x)}\right) \\ \iff x \in A \land y \in B \land y = e(x)
          \end{multline*}
    \item \textbf{Extensionality:}
          \begin{multline*}
            \left(\zlambda{A}{B}{x}{e(x)}\right) = \left(\zlambda{A}{B}{x}{e'(x)}\right) \\ \iff \forall\, x \in A,\, e(x) = e'(x)
          \end{multline*}
    \item \textbf{$\beta$-reduction:}
          \begin{equation*}
            \forall\, x \in A, \fapply{\left(\zlambda{A}{B}{x}{e(x)}\right)}{\langle x, \irrelprf \rangle} = e(x)
          \end{equation*}
    \item \textbf{$\eta$-expansion:}
          \begin{equation*}
            f = \left(\zlambda{A}{B}{x}{\fapply{f}{\langle x, \irrelprf \rangle}}\right)
          \end{equation*}
  \end{itemize}
  Each item is stated in Lean as displayed; \(\eta\)-expansion, for instance, reads
\leaninl{f = (λᶻ : A → B | x ↦ @ᶻf ⟨x, by zdom⟩)},
with \leaninl|zdom| supplying the proof written $\irrelprf$ above.
\end{theorem}

Application and abstraction integrate with rewriting, so that many proofs reduce to algebra on images, evaluation, and composition; the following example illustrates this.

\begin{example}
  \label[example]{ex:func-ext}
  Consider the extensionality principle for total functions: two functions $f,g\in B^A$ are equal if they are pointwise equal on $A$.
  In \zflean, this is stated as the following theorem:
\begin{leanlisting}
theorem is_func_ext_iff {A B : ZFSet} {f g : ZFSet}
    (hf : IsFunc A B f) (hg : IsFunc A B g) :
    f = g ↔ (∀ x ∈ A, @ᶻf ⟨x, by zdom⟩ = @ᶻg ⟨x, by zdom⟩)
\end{leanlisting}
  The two domain-membership obligations, \leaninl{x ∈ Dom(f)} and \leaninl{x ∈ Dom(g)}, are closed by \leaninl|zdom| from \leaninl{x ∈ A} and the hypotheses \leaninl{hf} and \leaninl{hg}, applied through lemma \leaninl{mem_dom_of_mem} (see \Cref{fig:automation-layer-thm}).
  A direct proof by extensionality on sets is tedious; it is easier to $\eta$-expand both $f$ and $g$ and apply the extensionality of \Cref{thm:lambda} to the resulting $\lambda$-abstractions.
  In \zflean, this is a rewrite with \leaninl|lambda_eta| for $f$ and $g$ and then with \leaninl|lambda_ext_iff|, which lifts extensionality for ZFC functions to the extensionality of Lean functions.
\end{example}

\subsection{Automation of side conditions}
\label{subsec:zflean-automation}
The semantic content of these side conditions is inherent: a domain membership is the definedness question itself (\Cref{subsec:limitations}). What our design adds is their explicit certificate form, the by-construction discipline of \Cref{rmk:no-junk}, which trades junk values for well-formedness witnesses instead of threading hypotheses through definitions.
Working with relations and partial/total functions repeatedly produces such subtype-shaped goals.
The library provides four lightweight tactics that try to solve them \emph{structurally}:
\begin{itemize}
  \item \leaninl{zrel}: relational goals $R\subseteq A\times B$;
  \item \leaninl{zpfun}: partial functionality goals $\IsPFunc(f,A,B)$;
  \item \leaninl{zfun}: total functionality goals $\IsFunc(A,B,f)$;
  \item \leaninl{zdom}: domain-membership goals $x \in \Dom(f)$.
\end{itemize}
All four tactics work similarly: they search among local hypotheses for witnesses of the required properties, or apply standard theorems (\eg composition of functions is a function) to reduce the goal to simpler subgoals.
The search backtracks when the application of a theorem fails; concretely, each tactic is a thin wrapper around Lean's \leaninl{solve_by_elim}, searching the local context together with dedicated lemma sets, and these sets nest: a relation goal may be closed from a partial- or total-function fact, so \leaninl{zrel} searches all three forms, \leaninl{zpfun} the last two, \leaninl{zfun} its own, and \leaninl{zdom} its own together with those of \leaninl{zfun} and \leaninl{zpfun}, since its seed lemmas carry $\IsFunc$ and $\IsPFunc$ hypotheses.
To reduce search time and keep automation usable at runtime, selected theorems are tagged with attributes, named \leaninl{zrel}, \leaninl{zpfun}, \leaninl{zfun}, and \leaninl{zdom} after the tactics, that form these lemma sets.
\Cref{fig:tactics} shows this design; tagging a new theorem extends the corresponding tactic.

This architecture is deliberately standard: attribute-indexed lemma sets driving a small search engine are how Mathlib packages domain-specific automation, over Aesop's extensible best-first search~\cite{aesop} (\leaninl|continuity|, \leaninl|measurability|) or over bespoke engines such as \leaninl|fun_prop|.
Depth-bounded backward chaining suffices because the side conditions are Horn-shaped closure goals, for which \leaninl|solve_by_elim| is fast and predictable; the contribution lies in the calculus and lemmas, not in the search engine.
\begin{figure}[tbp]
  \centering
  \Description{Diagram of the four automation tactics and the goal shapes they discharge. Dashed arrows go from zrel, zpfun, zfun, and zdom to their respective goal shapes (a relation goal, a partial-function goal, a total-function goal, and a domain-membership goal). Solid arrows show which tactic's lemma set feeds another: zpfun and zfun into zrel, zfun into zpfun, and zfun and zpfun into zdom.}
  \begin{tikzpicture}[
      inner sep=2.5pt,
      every node/.style={rounded corners=2pt},
      node distance=5mm and 2mm,
      arr/.style={-{stealth}, shorten <= 2pt, shorten >= 2pt},
    ]
    \scriptsize
    \node[draw] (rel_goal) {$\vdash R \subseteq A \times B$};
    \node[draw, left=of rel_goal] (pfun_goal) {$\vdash \IsPFunc(f,A,B)$};
    \node[draw, right=of rel_goal] (fun_goal) {$\vdash \IsFunc(A,B,f)$};
    \node[draw, right=of fun_goal] (dom_goal) {$\vdash x \in \Dom(f)$};
    \node[draw, below=of rel_goal, thick] (zrel) {\leaninl|zrel|};
    \node[draw, below left=of zrel, thick] (zpfun) {\leaninl|zpfun|};
    \node[draw, below right=of zrel, thick] (zfun) {\leaninl|zfun|};
    \node[draw, below=of dom_goal, thick, yshift=-9mm] (zdom) {\leaninl|zdom|};

    \draw[arr, dashed] (zrel) -- (rel_goal);
    \draw[arr, dashed] (zpfun) -| (pfun_goal);
    \draw[arr, dashed] (zfun) -| (fun_goal);
    \draw[arr, dashed] (zdom) -- (dom_goal);

    \draw[arr, thick] (zpfun) -- (zrel);
    \draw[arr, thick] (zfun) -- (zrel);
    \draw[arr, thick] (zfun) -- (zpfun);
    \draw[arr, thick] (zfun) to[bend right=10] (zdom);
    \draw[arr, thick] (zpfun.south east) to[bend right=10] (zdom.south west);
  \end{tikzpicture}
  \caption{Overview of the automation tactics for relations and functions. Dashed arrows indicate the shape of the goals discharged by each tactic; thick arrows indicate tactic invocations.}
  \label{fig:tactics}
\end{figure}
\begin{figure}[tbp]
\Description{Lean listing of four theorems tagged with automation attributes: subset_prod_inv tagged zrel, Id.IsPFunc tagged zpfun, IsFunc_of_composition_IsFunc tagged zfun, and mem_dom_of_mem tagged zdom.}
\begin{leanlisting}
@[zrel] -- converse of a relation is a relation
theorem subset_prod_inv {R A B : ZFSet}
  (hR : R ⊆ A.prod B) : R⁻¹ ⊆ B.prod A

@[zpfun] -- identity is a partial function
theorem Id.IsPFunc {A : ZFSet} : (𝟙A).IsPFunc A A

@[zfun] -- composition of total functions is one
theorem IsFunc_of_composition_IsFunc {g f : ZFSet}
  {A B C : ZFSet} (hg : B.IsFunc C g) (hf : A.IsFunc B f) :
    A.IsFunc C (composition g f A B C)

@[zdom] -- the domain of a function contains its source
theorem mem_dom_of_mem {f A B : ZFSet} (hf : IsFunc A B f)
  {x : ZFSet} (hx : x ∈ A) :
    x ∈ Dom f (is_rel_of_is_func hf)
\end{leanlisting}
\caption{Examples of theorems tagged for automation.}
\label{fig:automation-layer-thm}
\end{figure}
A few examples of such tagged theorems are shown in \Cref{fig:automation-layer-thm}.

Given this infrastructure, we then use Lean's \emph{automatic parameters} (\leaninl{autoParam}) to pass those tactics as default proof search strategies in partial definitions depending on side conditions, so that those side conditions are automatically discharged when possible.
Automatic parameters are similar to default parameters, but are specific to tactics: they try to discharge the goal using the provided default tactic if the argument is not supplied explicitly, and otherwise report the unsolved obligation as an error, leaving the user to provide a manual proof.
Furthermore, those side conditions often depend on parameters that are kept implicit (\eg the sets $A$ and $B$ in $R \subseteq A \times B$); they are also inferred automatically by Lean.
The inferred sets shape only the ambient product: two typings of one graph give the same domain, range, and converse, and the same applied value (\leaninl|dom_indep|, \leaninl|range_indep|, \leaninl|inv_indep|, \leaninl|fapply_val_indep|), so whichever typing the search finds first, the object is the same.
One can also pass explicit proofs instead of relying on automation, even when the tactics would succeed.
This lets one work with these definitions, most of the time, without attending to side conditions or implicit parameters.

\Cref{fig:automation-layer-thm} illustrates this already: the converse relation $R^{-1}$ in theorem \leaninl|subset_prod_inv| is a relation between $B$ and $A$, which are implicit at this level, and relies on the side condition $R \subseteq A \times B$, automatically discharged by \leaninl|zrel|.
\Cref{sec:zflean-case-study} shows this at work in a client.
We also tag relevant equality theorems with the \leaninl{simp} attribute, so that they are added to the simplifier's rewrite set: \leaninl|(𝟙A).Range = A|, \leaninl|R[∅] = ∅|, and \leaninl|f⁻¹[f[X]] = X| for a bijection $f$ and $X \subseteq A$, among others.

Overall, this calculus provides a compact, rewriting-friendly and extensible interface for extensional reasoning with relations and functions \emph{inside} Mathlib's model of ZFC, while hiding routine side conditions behind small, domain-specific tactics.

\section{Canonical Constructions and Universal Properties}
\label{sec:zflean-canon}

On top of the relational calculus of \Cref{sec:zflean-rel}, we build a library of canonical ZFC objects (Booleans, naturals, integers, rationals, sums, and options) and expose their characteristic properties in a form usable in practice.
The constructions themselves are textbook (we follow Kunen~\cite{kunen2011set} and claim no mathematical novelty); the contribution is their packaging: each object comes with its set-level anchor, its typed interface, and the eliminators, algebraic instances, and automation hooks with transfers that make them usable inside Lean.

\paragraph*{Set-level objects behind the typed interface}
We expose these objects through Lean \emph{types}, namely \leaninl{ZFBool}, \leaninl{ZFNat}, \leaninl{ZFInt}, \leaninl{ZFRat}, and each type has a \emph{set-level anchor}: a set of the model whose subtype of members is equivalent to the type.
Booleans and naturals are \emph{refinement types} over their carrier, \leaninl|ZFBool = {x : ZFSet // x ∈ 𝔹}| and \leaninl|ZFNat = {x : ZFSet // x ∈ ℕ}|: an element \leaninl|⟨e, h⟩| is a set paired with a proof of its membership, it can always be destructed to reason at the set level, and set-level results lift to the type.
Integers and rationals are more naturally \emph{quotients} and are exposed as quotient types; their anchors are a set of canonical representatives for $\mathbb{Z}$ and the set of equivalence classes for $\mathbb{Q}$, each with a proved equivalence to the quotient type.
In every case the typed layer is only an \emph{interface}, adopted for ergonomics rather than out of necessity: it lets us attach Lean's algebraic typeclasses and tactics (\eg \leaninl|ring|) and reuse Mathlib.
A statement about a refinement type is a statement about the members of its anchor; a statement about a quotient type transports along the equivalence to one about the members of its anchor (\Cref{tab:layers}).
\begin{table}[t]
  \centering
  \footnotesize
  \caption{The layers of the canonical constructions: the set-level anchor, the typed view over its members, and the native Lean type reached by the transfers of \Cref{subsec:bridges}.}
  \label{tab:layers}
  \setlength{\tabcolsep}{3pt}
  \begin{tabular}{llll}
    \hline
    & \textbf{Anchor (a set)} & \textbf{Typed view} & \textbf{Native} \\
    \hline
    Booleans  & $\mathbb{B}$ & \leaninl|ZFBool|, refinement & \leaninl|Bool| \\
    Naturals  & $\mathbb{N}$ & \leaninl|ZFNat|, refinement & \leaninl|ℕ| \\
    Integers  & canonical representatives & \leaninl|ZFInt|, quotient & \leaninl|ℤ| \\
    Rationals & the equivalence classes & \leaninl|ZFRat|, quotient & \leaninl|ℚ| \\
    Sums      & \leaninl|Sum.toZFSet A B| & $A \uplus B$, refinement & $\widehat{A} \oplus \widehat{B}$ \\
    Options   & \leaninl|Option.toZFSet A| & $\mathsf{Option}\,A$, refinement & $\mathsf{Option}\,\widehat{A}$ \\
    \hline
  \end{tabular}
\end{table}

\subsection{Canonical objects}
\paragraph*{Booleans}
The Boolean object is the two-element set $\mathbb{B} \triangleq \{\bot, \top\}$, with $\bot \triangleq \varnothing$ and $\top \triangleq \{\varnothing\}$; conjunction is intersection, disjunction union, and negation the complement relative to $\mathbb{B}$.
On the subtype \mbox{$\mathsf{ZFBool} \triangleq \{ x : \mathsf{ZFSet}\, \slash\mkern-3mu\slash\,  x \in \mathbb{B} \}$} (a type, not a set: the type of sets belonging to $\mathbb{B}$), a case eliminator is registered for Lean's \leaninl|cases| tactic ({\leaninl{ZFBool.casesOn}{}}, \leaninl|@[cases_eliminator]|), the connectives return members of $\mathbb{B}$ by construction ({\leaninl{ZFBool.and}{}}), and the expected laws are simplification lemmas.

\paragraph*{Natural numbers}
Mathlib already provides the first infinite von Neumann ordinal $\omega$ inside the model, as the image of Lean's naturals (\leaninl|PSet.omega| maps $n$ to the $n$-th von Neumann ordinal); we nonetheless re-derive $\mathbb{N}$ from the axiom of infinity as the least inductive set, and prove the two equal (\Cref{rmk:nat-agnostic}), because the construction fixes the successor, recursion and induction, and the arithmetic with proved computation rules.
A set $X$ is \emph{inductive} ({\leaninl{ZFSet.inductive_set}{}}) when $\varnothing \in X$ and $n \cup \{n\} \in X$ for every $n \in X$.
The axiom of infinity guarantees the existence of an inductive set (in the artifact, the existential is supplied by Mathlib's $\omega$), of which we classically fix an abstract witness $S^{\infty}$; the definition below consumes only the existential, and \Cref{rmk:nat-agnostic} makes it independent of the witness.
\begin{definition}[Natural numbers]
  \label[definition]{def:nat}
  $\mathbb{N}$ is the smallest inductive set:
  \[
    \mathbb{N} \triangleq \bigcap \setcomp{X \subseteq S^{\infty}}{X\ \text{is inductive}}
  \]
\end{definition}
\begin{remark}
  \label[remark]{rmk:nat-agnostic}
  \Cref{def:nat} is agnostic to the choice of $S^{\infty}$: every inductive set contains $\mathbb{N}$ ({\leaninl{Nat_subset_of_inductive}{}}), so any two witnesses yield the same least inductive set ({\leaninl{Nat_eq_of_inductive}{}}), and that set is $\omega$ itself ({\leaninl{Nat_eq_omega}{}}): the two notions coincide, and Mathlib's lemmas about $\omega$ apply to the elements of $\mathbb{N}$.
  This is the invariance property that discharges the one arbitrary selection of \Cref{rmk:choice}.
\end{remark}
The naturals are exposed as the subtype $\mathsf{ZFNat} \triangleq$ \leaninl|{n : ZFSet // n ∈ ℕ}|, whose elements are pairs $\langle n, h\rangle$ with $h : n \in \mathbb{N}$; the proof is elided below, so that $0_{\mathbb{N}}$ denotes the element with underlying set $\varnothing$, and $\mathsf{succ}\,\langle n, h \rangle \triangleq \langle n \cup \{n\}, h' \rangle$, where $h'$ is derived from $h$.
Each natural is the set of the smaller ones; the strict order is membership on the underlying sets, giving irreflexivity, transitivity, and trichotomy, with $m \leq n \triangleq m < n \lor m = n$.
Since $\mathsf{succ}$ induces a well-founded relation on $\mathsf{ZFNat}$, we derive a recursion principle, with weak and strong induction, as a fixpoint; the attribute registers it as the eliminator of Lean's \leaninl|induction| tactic:
\begin{leanlisting}
@[induction_eliminator]
def rec.{u} {motive : ZFNat → Sort u} (n : ZFNat)
  (zero : motive 0)
  (succ : Π x, motive x → motive (succ x)) : motive n
\end{leanlisting}
Arithmetic is defined by primitive recursion on this eliminator: addition folds the successor over its second argument, \leaninl|ZFNat.rec n m (fun _ : ZFNat => succ)|, multiplication folds addition, and subtraction folds the set-level predecessor $\mathsf{pred}(n) = \bigcup n$.
A library of arithmetic lemmas culminates in a \emph{commutative semiring} instance, which makes Mathlib's \leaninl|ring| tactic (an automatic prover for equalities in commutative (semi)rings by normalization) available on $\mathsf{ZFNat}$ goals:
\begin{leanlisting}
example (a b : ZFNat) :
  (a+b)^2 = a^2 + 2*a*b + b^2 := by ring
\end{leanlisting}

\paragraph*{Integers and rationals}
\newcommand{\zrel}{\ensuremath{\sim_{\mathbb{Z}}}}
Following Kunen~\cite[Ch.~1, Sec.~11]{kunen2011set}, integers are equivalence classes of pairs of naturals under the relation $(a,b) \zrel (c,d) \iff a + d = b + c$.
Contrary to naturals, we expose them primarily as a \emph{quotient} type, which gives the cleanest interface for lifting operations, without abandoning the set level: a genuine set $\mathbb{Z}$ of canonical representatives is built alongside and proved isomorphic to the quotient.
\begin{definition}[Ring of integers]
  \[
    \mathsf{ZFInt} \triangleq (\mathsf{ZFNat} \times \mathsf{ZFNat})/\!\zrel
    ,\qquad
    \mathbb{Z} \triangleq \mathbb{N} \times \{0_{\mathbb{N}}\}\, \cup\, \{0_{\mathbb{N}}\} \times \mathbb{N}
  \]
\end{definition}
An integer is a ``distance'' between two naturals: the opposite is the pair flip, and $\mathbb{N}$ is isomorphic to the nonnegative part of $\mathbb{Z}$.
The constants and the arithmetic operations are lifted from representative pairs, each lift carrying a proof that it respects $\zrel$, and $\mathsf{ZFInt}$ receives a \emph{commutative ring} instance.
Its total order is given directly on representatives, $[(a,b)]_{\zrel} \le [(c,d)]_{\zrel} \iff a + d \le b + c$, rather than transported through any equivalence.
The subtype built from the set $\mathbb{Z}$ is proved isomorphic to $\mathsf{ZFInt}$ ({\leaninl{instEquivZFIntInt}{}}), which also provides coercions between the two.

\newcommand{\qrel}{\ensuremath{\sim_{\mathbb{Q}}}}
Rationals follow the same pattern: $\mathsf{ZFRat} \triangleq (\mathsf{ZFInt} \times \mathsf{ZFInt}^{\star})/\!\qrel$ with $(a,b) \qrel (c,d) \iff a * d = b * c$, where $\mathsf{ZFInt}^{\star}$ is $\mathsf{ZFInt}$ without $0_{\mathbb{Z}}$, so that the nonzero-denominator constraint is carried by the subtype rather than threaded as a side condition.
Addition, subtraction, and multiplication are lifted from $\mathsf{ZFInt}$, the multiplicative inverse swaps numerator and denominator on nonzero classes, and we prove the remaining axioms to instantiate a \emph{field} structure on $\mathsf{ZFRat}$.
The order is where the quotient bites: denominators are arbitrary nonzero integers, so naive cross-multiplication is sign-sensitive and does not lift.
Positivity of a class represented by $(a, b)$ is instead defined as $0 < a * b$ in $\mathsf{ZFInt}$ (\leaninl|isPos|), and $x < y$ as positivity of $y - x$; positivity is independent of the representative, since $(a, b) \qrel (c, d)$ gives $(a * b)(d * d) = (b * b)(c * d)$ with the two squares positive, so $a * b$ and $c * d$ share a sign.
Trichotomy and the closure of the positive cone under addition and multiplication are proved on representatives and instantiate \leaninl|LinearOrder| and \leaninl|IsOrderedRing| for $\mathsf{ZFRat}$; the Archimedean property is the transfer example of \Cref{subsec:bridges}.
The anchor of the rationals is built by a generic \emph{set-level quotient}: for a set $E$, writing $\widehat{E} \coloneq \{ x : \mathsf{ZFSet} \,\slash\mkern-3mu\slash\, x \in E \}$, and an equivalence relation $R \subseteq E \times E$, $E/R \triangleq \setcomp{c \in \mathcal{P}(E)}{\exists x \in c,\ \forall y,\ y \in c \iff (x,y) \in R}$ is the set of equivalence classes ({\leaninl{ZFQuotient}{}}), and the subtype of its members is equivalent to Lean's quotient of $\widehat{E}$ by the relation read off $R$ ({\leaninl{equivZFQuotient}{}}).
Applied to $\mathbb{Z} \times (\mathbb{Z} \setminus \{0\})$ and to $\qrel$ written as a set, it gives $\mathbb{Q}$ as a set with $\widehat{\mathbb{Q}} \simeq \mathsf{ZFRat}$ ({\leaninl{equivRatZFRat}{}}); the equivalence between relations on $E$ as sets and Lean relations on $\widehat{E}$ that this uses ({\leaninl{equivZFRelation}{}}) is itself registered for the \leaninl|transfer| tactic of \Cref{subsec:bridges}.

\paragraph*{Coproducts and options}
The coproduct of two sets tags its elements with a Boolean ({\leaninl{ZFSet.Sum}{}}): $A \uplus B \triangleq \{ x : \mathsf{ZFSet} \,\slash\mkern-3mu\slash\, x \in (\{ \bot \} \times A) \cup (\{ \top \} \times B) \}$, with injections $\mathsf{inl} : A \to A \uplus B$ and $\mathsf{inr} : B \to A \uplus B$ and a case eliminator; the option type is $\mathsf{Option}\, A \triangleq \{ \varnothing \} \uplus A$, with $\mathsf{none} \triangleq \mathsf{inl}(\varnothing)$ and $\mathsf{some} \triangleq \mathsf{inr}$; $A \uplus B$ names the \emph{type}, and \leaninl|Sum.toZFSet| its carrier set, the anchor of \Cref{tab:layers}.
They carry less mathematical weight than the previous objects but model sum types and an alternative form of partiality.

The universal properties of $\mathbb{N}$ and $A \uplus B$ are stated at the set level, with arrows in $B^A$ and relational composition: $(\mathbb{N}, \varnothing, \mathsf{succ})$ is a natural numbers object ({\leaninl{funs_nat_recursion}{}}), and two functions out of $A$ and $B$ factor uniquely through $A \uplus B$ ({\leaninl{funs_coprod_universal}{}}).

\section{Embeddings, isomorphisms, and transfers}
\label{subsec:iso-bridges}
\subsection{Embeddings and isomorphisms}
We define embeddings and isomorphisms between sets at the level of the model's sets.
\begin{definition}[Embeddings, isomorphisms]
  A set $A$ is \emph{embeddable} into a set $B$, denoted $A \zbinop{\hookrightarrow} B$, when there exists an injective total function from $A$ to $B$:
  \[
    A \zbinop{\hookrightarrow} B \;\triangleq\; \exists\, f\, (h_f : \IsFunc(A,B,f)),\, \IsInjective(f)
  \]
  $A$ and $B$ are said to be \emph{isomorphic}, denoted $A \zbinop{\cong} B$, when there exists a bijective function\footnote{In set theory, or more precisely in the category of sets, isomorphisms are exactly the bijections. This definition therefore captures the usual notion of isomorphism between sets.} from $A$ to $B$:
  \[
    A \zbinop{\cong} B \;\triangleq\; \exists\, f\, (h_f : \IsFunc(A,B,f)),\, \IsBijective(f)
  \]
  In both cases $f$ ranges over sets ($f : \mathsf{ZFSet}$ in the artifact).
\end{definition}
\begin{remark}
  $\IsInjective(f)$ and $\IsBijective(f)$ expect $f$ to be a total function from $A$ to $B$, which is given by $h_f$; in practice, this side condition is handled by our automation layer, here the \leaninl|zfun| tactic.
\end{remark}
The relation $\zbinop{\cong}$ is shown to be an equivalence relation on sets, and $\zbinop{\hookrightarrow}$ is shown to be antisymmetric up to isomorphism, a property known as the \emph{Cantor--Schr\"{o}der--Bernstein theorem}~\cite{CSB_thm}.
Its \zflean statement takes the two embeddings in destructed form, the underlying injections rather than the existential packages:
\begin{leanlisting}
theorem isIso_of_biembedding {E F f g : ZFSet}
  {hf : E.IsFunc F f} (f_inj : IsInjective f hf)
  {hg : F.IsFunc E g} (g_inj : IsInjective g hg) : E ≅ᶻ F
\end{leanlisting}
\begin{theorem}[Cantor--Schr\"{o}der--Bernstein]
  \label{thm:csb}
  For any sets \(A\) and \(B\), if \(A\) can be embedded into \(B\) and \(B\) can be embedded into \(A\), then \(A\) and \(B\) are isomorphic.
\end{theorem}
\begin{proof}[Proof sketch]
  Let $f \colon A \to B$ and $g \colon B \to A$ be the injections, $S_0 \triangleq A \setminus g(B)$, $S_{n+1} \triangleq g(f(S_n))$, and $S^* \triangleq \bigcup_{n} S_n$, a set closed under $g \circ f$.
  The bijection sends $x \in S^*$ to $f(x)$ and $x \in A \setminus S^*$ to $g^{-1}(x)$, defined since $A \setminus S^* \subseteq g(B)$; it is injective because $f$ and $g$ are and $f(S^*)$ is disjoint from $B \setminus f(S^*)$, and surjective because these two sets cover $B$.
\end{proof}

\begin{remark}
  This theorem exercises the relational calculus: the back-and-forth construction combines the library's images, composition, and inverses with separation, difference, and a family of iterates indexed by Lean's naturals, while the recurring well-formedness side conditions are discharged automatically.
\end{remark}

\subsection{Transfers to Lean types and the \texorpdfstring{\texttt{transfer}}{transfer} tactic}
\label{subsec:bridges}
To interoperate with Lean/Mathlib infrastructure, our library provides explicit \emph{transfers} between canonical ZFC objects and native types: we prove that our constructions are isomorphic to their Lean counterparts, here at the level of \emph{types}, denoted $\simeq$ (Lean's \leaninl{Equiv}, a pair of mutually inverse functions between the two types), as opposed to the set-level isomorphism $\zbinop{\cong}$ between ZFC sets defined above.
Writing $\widehat{A}$ for the subtype $\{ x : \mathsf{ZFSet} \,\slash\mkern-3mu\slash\, x \in A \}$ of a set $A$, the transfers read
$\mathsf{ZFBool} \simeq \mathsf{Bool}$,
$\mathsf{ZFNat} \simeq_{+*} \mathsf{Nat}$,
$\mathsf{ZFInt} \simeq_{+*} \mathsf{Int}$,
$\mathsf{ZFRat} \simeq_{+*} \mathsf{Rat}$,
$A \uplus B \simeq \widehat{A} \oplus \widehat{B}$, and
$\mathsf{ZFSet.Option}\, A \simeq \mathsf{Option}\, \widehat{A}$,
where $\mathsf{Nat}$, $\mathsf{Int}$, $\mathsf{Rat}$, $\oplus$, and the second $\mathsf{Option}$ are Lean's native types (\leaninl|ℕ|, \leaninl|ℤ|, \leaninl|ℚ|, \leaninl|⊕|, and \leaninl|Option| in Lean code), and $\simeq_{+*}$ is a \emph{ring} equivalence, preserving addition and multiplication ({\leaninl{ringEquivNat}{}}, {\leaninl{equivInt}{}}, {\leaninl{equivRat}{}}).
The transfer for the rationals is obtained from the cast $\mathsf{Rat} \to \mathsf{ZFRat}$, a ring morphism that is injective because $\mathsf{ZFRat}$ has characteristic zero and surjective because every class $[(a,b)]_{\qrel}$ is the quotient of the images of two integers.

Such an isomorphism lets us \emph{transport} statements and proofs across the boundary, and the library packages the transport as a tactic: \leaninl|transfer α → β => tacs| rewrites the goal about $\alpha$ into the equivalent goal about $\beta$ and runs \leaninl|tacs| on it.
The equivalence is found by instance synthesis on a class \leaninl|TransferEquiv α β|, registered for $\mathsf{ZFBool}$, $\mathsf{ZFNat}$, $\mathsf{ZFInt}$, and $\mathsf{ZFRat}$, or is passed explicitly as \leaninl|transfer α → β using e => tacs|; a bundled isomorphism ($\simeq_{+*}$, $\simeq_{o}$, \ldots) is accepted wherever an equivalence is.
The tactic first reverts every hypothesis whose type mentions $\alpha$ and replaces each binder \leaninl|∀ (x : t)| whose type is built from $\alpha$ by arrows and products with a binder over $t[\beta/\alpha]$, through Mathlib's \leaninl|Equiv.forall_congr_left| and an equivalence $t \simeq t[\beta/\alpha]$ assembled from $e$ by congruence; a hypothesis \leaninl|f : ZFNat → ZFNat → ZFNat| thus comes back as \leaninl|f : ℕ → ℕ → ℕ|.
It then runs \leaninl|simp| with a dedicated simp set, \leaninl|transfer_simps|, whose lemmas state how the operations of $\alpha$ are read in $\beta$, such as \leaninl|e (a + b) = e a + e b|, together with generic lemmas lifting equalities, $\forall$, and $\exists$ along the equivalences at hand, until no $\alpha$ is left.
It finally reintroduces the hypotheses under their original names and runs \leaninl|tacs|.
The goal is rewritten in place: closing the transferred goal closes the original one, so nothing is transferred back.
Two uses from the artifact's example file follow; the comments give the goal each block receives.
\begin{leanlisting}
example (x : ZFRat) : ∃ n : ℕ, x < n := by
  transfer ZFRat → ℚ =>     -- x : ℚ ⊢ ∃ n, x < ↑n
    exact exists_nat_gt x

example (p : ZFBool) : p ⋁ p.not = ⊤ := by
  transfer ZFBool → Prop using ZFBool.equivProp =>
    exact em p              -- p : Prop ⊢ p ∨ ¬p
\end{leanlisting}
In the first, $\mathsf{ZFRat}$ is Archimedean because $\mathsf{Rat}$ is: $n$ ranges over $\mathsf{Nat}$, a type the transfer leaves alone, and the cast $(n : \mathsf{ZFRat})$ becomes the cast $(n : \mathsf{Rat})$.
In the second, $\mathsf{ZFBool}$ is read as $\mathsf{Prop}$ rather than as $\mathsf{Bool}$, through the classical equivalence \leaninl|equivProp|, named because it is not the registered instance: an equation between Booleans becomes an equivalence of propositions, $\top$ becomes $\mathsf{True}$, and the goal is excluded middle, which is where the classical equivalence is paid for.

Beyond the transfer itself, a new type needs an instance and the lemmas that no generic lemma covers.
When the transfer is a bundled isomorphism, Mathlib's \leaninl|map_add|, \leaninl|map_mul|, \leaninl|map_natCast|, and their kin, which belong to \leaninl|transfer_simps|, already describe how the ring operations, numerals, and casts travel; registering the instance is then enough for them.
What no such lemma covers is stated once per type, one line each: the relations (\leaninl|a ≤ b ↔ e a ≤ e b|, since no \leaninl|map_| lemma speaks about $\le$) and the operations outside the algebraic structure.
$\mathsf{ZFNat}$ declares six such lemmas (successor, truncated subtraction, power, divisibility, and the two orders), $\mathsf{ZFInt}$ ten, and $\mathsf{ZFRat}$ three; $\mathsf{ZFBool}$ carries no algebraic structure and states its two constants and three connectives directly, once per reading.
Congruence covers arrows and products; a binder of another shape, a \leaninl|List ZFNat| or a dependent $\forall$, is left as is, and the block then faces a goal that still mentions the source typed view.
The supported fragment is thus the goals over the typed interface: the types with a registered equivalence, closed under arrows and products, whose operations all have a \leaninl|transfer_simps| lemma. Raw \leaninl|ZFSet| propositions such as $x \in A$ and the function spaces $B^A$ of \Cref{sec:zflean-rel} are outside it; the direction is always the one written in the invocation, and failure is visible, as a residual goal that still mentions the source type.
The relation lemmas are also oriented, from the typed views toward Mathlib's types: the tactic accepts the inverse isomorphism, and Mathlib's lemmas carry the ring operations along it, but an order relation would need its lemmas restated for that direction.

The transfers also carry algebraic \emph{structure}, not only statements: the Euclidean-division development on $\mathsf{ZFInt}$ obtains its \leaninl|EuclideanDomain| instances by transport along the ring equivalence above, and then exposes the resulting division and modulo operations ({\leaninl{ediv}{}}, {\leaninl{emod}{}}) on the typed view of the integers.

\section{Case Study and Evaluation}
\label{sec:zflean-case-study}

We evaluate \zflean on a client development written after the library, a denotational semantics of a small imperative language whose stores, expressions, and commands denote sets of the model, then report which of its side conditions the tactics discharged, and the size and elaboration time of the artifact.

\subsection{A denotational semantics on set-level stores}
\label{subsec:imp}

Let $V$ be a set of variable names.
A store is a partial function from $V$ to $\mathbb{N}$, and the state space is the \emph{set} of all stores, $\mathsf{Store}(V) \triangleq \setcomp{\sigma \in \mathcal{P}(V \times \mathbb{N})}{\IsPFunc(\sigma, V, \mathbb{N})}$, so that a store, its domain, and the set of all stores are objects of one universe, as in B~\cite{AbrialBBook}.
The syntax is a Lean inductive type over $V$ and the denotations are the sets of \Cref{fig:imp-sem}: an expression denotes a partial function $\mathsf{Store}(V) \rightharpoonup \mathbb{N}$, a comprehension over $\mathsf{Store}(V) \times \mathbb{N}$, and a command a relation on $\mathsf{Store}(V)$, with $\mathtt{skip}$ the identity, sequencing relational composition, and the conditional the union of its branches restricted ($S \mathbin{\triangleleft} R$, B's domain restriction) to the sets $t_e$ and $f_e$ of stores on which the test is nonzero, respectively zero.
The loop is the union of the chain of its iterates,
\begin{gather*}
  \csem{\mathtt{while}\ e\ \mathtt{do}\ c} \triangleq \textstyle\bigcup_{n} W_n, \qquad W_0 \triangleq \varnothing,\\
  W_{n+1} \triangleq (f_e \mathbin{\triangleleft} \Id{\mathsf{Store}(V)}) \cup (t_e \mathbin{\triangleleft} (W_n \circ \csem{c})),
\end{gather*}
a set built by separation (\leaninl|loop|, \leaninl|loopIter|), and the assignment $x := e$ is the abstraction
\[
  \zlambda{\Dom(\sem{e})}{\mathsf{Store}(V)}{h : \sigma}{\sigma[x \mapsto \fapply{\sem{e}}{\langle \sigma, h \rangle}]}
\]
where $\sigma[x \mapsto n]$ is B's override: it is defined exactly on the stores on which $e$ is, and the domain proof that $\fapplysym$ requires is the membership proof $h$ of the binder itself (\Cref{subsec:zflean-func-eval-lambda}).
\begin{figure}[tbp]
\Description{Lean listing of the case study: the set Store V of stores, the denotation of expressions by cases on literals, variables and sums, and the denotation of commands by cases on skip, assignment, sequencing and conditional.}
\begin{leanlisting}
@[irreducible] def Store (V : ZFSet) : ZFSet :=
  (V.prod Nat).powerset.sep fun σ => σ.IsPFunc V Nat
def Expr.sem : Expr V → ZFSet
  | lit n => λᶻ : Store V → Nat | _σ ↦ n.val
  | var x => ((Store V).prod Nat).sep fun p =>
      ∃ σ n : ZFSet, p = σ.pair n ∧ x.val.pair n ∈ σ
  | add e₁ e₂ => ((Store V).prod Nat).sep fun p =>
      ∃ σ : ZFSet, ∃ n₁ n₂ : ZFNat, p = σ.pair (n₁ + n₂).val
        ∧ σ.pair n₁.val ∈ e₁.sem ∧ σ.pair n₂.val ∈ e₂.sem
def Cmd.sem : Cmd V → ZFSet
  | skip => 𝟙(Store V)
  | assign x e => λᶻ : ⟦e⟧ₑ.Dom → Store V
      | h : σ ↦ σ[x.val ↦ (@ᶻ⟦e⟧ₑ ⟨σ, h⟩).val]
  | seq c₁ c₂ =>
      composition c₂.sem c₁.sem (Store V) (Store V) (Store V)
  | ite e c₁ c₂ => (e.tt ◁ c₁.sem) ∪ (e.ff ◁ c₂.sem)
  | whileDo e c => loop (Store V) e.tt e.ff c.sem
\end{leanlisting}
\caption{The case study's definitions: the state space and the denotations $\sem{e}$ and $\csem{c}$; \leaninl|e.tt| and \leaninl|e.ff| are the sets $t_e$ and $f_e$ of stores on which $e$ is defined and nonzero, respectively zero.}
\label{fig:imp-sem}
\end{figure}

\begin{theorem}[Definedness]
  \label{thm:imp-defined}
  For every store $\sigma \in \mathsf{Store}(V)$ and expression $e$, $\sigma \in \Dom(\sem{e})$ if and only if $\mathrm{FV}(e) \subseteq \Dom(\sigma)$, where $\mathrm{FV}(e) \subseteq V$ is the set of variables of $e$ (\leaninl|Expr.mem_dom_sem_iff|).
\end{theorem}
This is B's well-definedness condition for an expression, here a theorem about domains.
\begin{theorem}[Determinism]
  \label{thm:imp-det}
  For every command $c$, $\csem{c}$ is a partial function from $\mathsf{Store}(V)$ to $\mathsf{Store}(V)$ (\leaninl|Cmd.sem_pfunc|), and an assignment $x := e$ is a total function on $\Dom(\sem{e})$ (\leaninl|sem_assign_isFunc|).
\end{theorem}
The loop's iterates form a chain of partial functions with disjoint stop and continue sets, so their union is one (\leaninl|loop_pfunc|); the loop satisfies its unfolding law $\csem{\mathtt{while}\ e\ \mathtt{do}\ c} = \csem{\mathtt{if}\ e\ \mathtt{then}\ (c; \mathtt{while}\ e\ \mathtt{do}\ c)\ \mathtt{else}\ \mathtt{skip}}$ (\leaninl|sem_whileDo_unfold|), and is the least relation closed under its stop and step rules (\leaninl|sem_whileDo_least|).
Program equivalences are equalities of sets: $\csem{\mathtt{skip}; c} = \csem{c} = \csem{c; \mathtt{skip}}$ and $\csem{(c_1; c_2); c_3} = \csem{c_1; (c_2; c_3)}$ are the identity and associativity laws of relational composition (\leaninl|Id.composition_left|, \leaninl|composition_assoc|), one rewrite each, with no explicit quotient or extensionality step in the client proof.
The assignment $x := x + 1$ increases the value of $x$ (\leaninl|assign_incr|): the proof unfolds the abstraction, reads the value off the denotation of the sum, and is left with $n < n + 1$ in $\mathsf{ZFNat}$, which \leaninl|transfer ZFNat → ℕ => omega| closes.

The development, a separate build target that imports the library as a downstream project would, is 457 lines and 63 theorems; its 204-line prelude adds what the library lacked (domain restriction, override, disjoint unions of partial functions, the loop operator) with closure lemmas tagged for the tactics of \Cref{subsec:zflean-automation}, which is how a client extends them.

\subsection{Side conditions: what the tactics discharged}
\label{subsec:obligations}

The development generates or exposes 30 leaf obligations, goals handed to a tactic or exposed by the one hand-applied rule: the 9 relation obligations of $\Dom$ parameters (\leaninl|zrel|), the 6 partial-functionality obligations of $\fapplysym$ parameters and the 7 partial-functionality goals in the proof of \Cref{thm:imp-det} (\leaninl|zpfun|), the 4 totality obligations of $\fcompZ$ parameters (\leaninl|zfun|), and the 4 domain proofs written \leaninl|by zdom|.
\Cref{tab:obligations} classifies them by the work left to the user; 22 of the 26 one-invocation sites and two of the three exposed premises close only with seeds the client tagged itself.
The hand-applied rule is the conditional's union rule: its disjointness premise is a five-binder implication, and \leaninl|solve_by_elim| exhausts its default depth of six introducing it; \leaninl|zpfun| closes the three premises once the rule is applied.
The manual proof is the assignment case: its abstraction is a partial function on $\Dom(\sem{e})$ where the goal asks for one on $\mathsf{Store}(V)$, and the closure search has no weakening rule; one application of \leaninl|pfunc_weaken| supplies the step.
To isolate what the tagged rules contribute, the lower half of \Cref{tab:obligations} reruns the same 30 obligations under plain \leaninl|solve_by_elim| with local hypotheses only, under the library's seeds with every client-tagged rule erased, and in the release environment.
The three hypothesis-shaped obligations survive with no seeds at all; the library's core adds the identity and composition cases; the client's seeds carry the remaining 24.
Those seeds are 20 lemmas (7 for \leaninl|zrel|, 9 for \leaninl|zpfun|, 3 for \leaninl|zfun|, 1 for \leaninl|zdom|), 115 lines with their statements; ten of them head the 24 proofs, the relation lemma of expression denotations and the totality lemma of literal assignments five each.
The harness (\texttt{casestudy/ImpObligations.lean}, a separate target) restates the 30 obligations; its build asserts every claimed success and failure.
\begin{table}[tbp]
  \centering
  \small
  \caption{The case study's 30 leaf obligations, by the work left to the user (top) and by cumulative seed configuration (bottom).}
  \label{tab:obligations}
  \setlength{\tabcolsep}{5pt}
  \begin{tabular}{lr}
    \hline
    closed by one tactic invocation                          & 26 \\
    closed after a rule applied by hand (union, conditional) & 3  \\
    manual proof (weakening, assignment)                     & 1  \\
    \hline
    closed by plain \leaninl|solve_by_elim|, no seeds        & 3  \\
    closed by the library's seeds only                       & 5  \\
    closed by library and client seeds (release)             & 29 \\
    \hline
  \end{tabular}
\end{table}

On the library itself, \leaninl|zdom| is measured on its 74 application sites, counted by \texttt{count\_sites.sh}; it closes 73 of them, eight of these with six membership seeds: the two elimination forms of $(a,b) \in A \times B$ through the projections and the four membership rules of the coproduct carrier of \Cref{sec:zflean-canon}.
The remaining site, in the Cantor--Schr\"oder--Bernstein proof, destructures an existential hypothesis, which an apply-only search does not do.

The prelude also shows what a development on Mathlib's \leaninl|ZFSet| alone writes by hand: the partial-functionality lemma of the client's own override is 11 lines, of the same shape as the library's 19-line lemma for composition, which sequencing then inherits through \leaninl|zpfun|, and as the 15-line identity law behind $\csem{\mathtt{skip}; c} = \csem{c}$.

\subsection{Size and elaboration time}
\label{subsec:artifacts}

\begin{table}[tbp]
  \centering
  \footnotesize
  \caption{Library size per component, measured by the script \texttt{count\_loc.py}; elaboration time of the modules (median of three clean rebuilds per module on an Apple M2 laptop, \texttt{time\_components.sh}). The total sums the per-module times; a clean build of the whole artifact is 73 s wall-clock (Lean 4.33.0, Mathlib v4.33.0).}
  \label{tab:artifacts-loc}
  \setlength{\tabcolsep}{4pt}
  \begin{tabular}{lrrr}
    \hline
    \textbf{Component}                           & \textbf{LoC}    & \textbf{\#Thms} & \textbf{Time (s)} \\
    \hline
    Relational calculus                          & 2,626           & 147             & 29.2 \\
    Embeddings, isomorphisms                     & 1,610           & 27              & 18.7 \\
    Naturals, set-level recursion                & 1,445           & 224             & 6.5  \\
    Integers, integer division                   & 1,340           & 156             & 3.3  \\
    Booleans, sums, quotients, rationals         & 2,220           & 221             & 20.5 \\
    Transfer, lemma sets, examples               & 311             & 8               & 6.5  \\
    Core glue and automation                     & 405             & 51              & 6.2  \\
    \hline
    \textbf{Total}                               & \textbf{9,957}  & \textbf{834}    & \textbf{91.0} \\
    \hline
  \end{tabular}
\end{table}

The library is self-contained on top of Mathlib's ZFC model: 9,957 lines of code and 834 theorems in the modules of \Cref{tab:artifacts-loc}; the case-study client, a separate target of the same artifact, adds 457 lines and 63 theorems, measured the same way, and elaborates in about ten seconds under the same three-run protocol.
Lines are counted excluding blank lines and comments, line and block; theorems count \leaninl|theorem| and \leaninl|lemma| declarations; both counts, the \leaninl|zdom| measurement, and the timings are reproduced by scripts shipped with the artifact.
The ``transfer'' component is the tactic of \Cref{subsec:bridges}, its generic lemma sets, and the example file; ``core glue and automation'' holds the shared base definitions, the attribute declarations, and the tactic front ends of \Cref{subsec:zflean-automation}.

\subsection{Discussion and limitations}
\label{subsec:limitations}

The automation is a narrow, bounded closure search over tagged lemmas and local hypotheses rather than general-purpose proof search, and \Cref{subsec:obligations} shows where it stops: weakening, destructuring, and the depth of a premise.
Domain-membership obligations at application sites ($x \in \Dom(f)$) stay explicit arguments of $\fapplysym$ by design: a domain obligation is the definedness question itself, the same information that PVS surfaces as type-correctness conditions~\cite{pvs}, and, as there, it is discharged by a standard strategy when possible and exposed to the user otherwise.
Decision procedures hard-wired to Lean's numeric types, \eg \leaninl|omega|, do not act on the typed views of their set-theoretic constructions; the \leaninl|transfer| tactic of \Cref{subsec:bridges} brings the goal to them, as in \leaninl|assign_incr| above.
Outside the library, Tr\'elat's BEer~\cite{beer} draws on the function spaces and $\lamZ$ of \Cref{sec:zflean-rel}, the option encoding, and the Boolean and integer constructions of \Cref{sec:zflean-canon}; Bergeron et~al.'s PlusCal semantics~\cite{fugue} is at an earlier stage, and its footprint in the library is not yet settled.
Three lessons carry over: refinement views with a set-level anchor transfer more easily than quotient views, so quotients only where the classes are the object, as for $\mathbb{Q}$; closure automation grows by tagging seeds rather than by deepening the search, and stops at weakening and destructuring; and explicit domain certificates are the price of a junk-free interface, repaid by standard-strategy discharge.

\section{Related Work}
\label{sec:zflean-related}
We position our contributions with respect to set-theoretic developments in interactive provers such as Isabelle, Rocq, Mizar, and Metamath, to prior work on Mathlib's ZFC model, and to categorical foundations.

\paragraph*{Isabelle/ZF and ZFC in Isabelle/HOL}
Isabelle/ZF~\cite{paulson-zf} develops ZF as an object-level theory of a first-order logic, with membership and extensionality primitive, up to the constructible universe and relative consistency results, but it lives apart from Isabelle/HOL, so HOL's automation, type classes, and libraries are not directly available.
Zhan's \textsf{auto2}~\cite{zhan-auto2} carries large set-theoretic developments in Isabelle/FOL, up to the fundamental group, with a strong specialized prover; \zflean's tactics stay narrow, and general proof effort falls back on Lean's.
Encoding ZFC \emph{within} Isabelle/HOL, as a HOL type of sets with an elementhood relation, buys close integration with HOL's automation and type classes~\cite{paulson-zf-hol}; \zflean bridges the same way, keeping ZFC-level objects and results with transfers to Lean's native types, and adds a relational calculus tailored to set-level developments where the HOL encodings keep a minimal notation surface.

\paragraph*{Rocq (ZFC, IZF)}
Rocq has set-theoretic developments with different aims: encodings of ZF à la Aczel~\cite{aczel-set-theory}, with primitive membership and extensional equality, in which the ZF axioms are theorems and ZFC follows from choice in the host~\cite{werner-sets-types} (the construction Mathlib's model itself follows), and constructive or intuitionistic set theories aimed at constructive metatheory or models of type theory, through pointed-graph semantics~\cite{miquel-pts-set-theory} or Grothendieck universes~\cite{barras-set-theory}.
They show that set-level mathematics is possible inside a dependent type theory; integration with the typed algebraic hierarchy and automation of the host remains ad hoc.

\paragraph*{Transfer and transport}
In Isabelle/HOL, Lifting and Transfer~\cite{lifting-transfer} derive transfer rules between a raw and an abstract type and discharge goals by structural decomposition along them; Types-to-Sets~\cite{types-to-sets} relativizes type-based theorems to set-based ones through local type definitions; Kappelmann's transport framework~\cite{kappelmann-transport} generalizes both via partial Galois connections; and in Rocq, Trocq~\cite{trocq} transfers along type equivalences and weaker relations by parametricity.
All operate between two formulations within the host logic, and all but Types-to-Sets accept relations weaker than equivalences.
Our \leaninl|transfer| is narrower on both counts: it relates the typed views over \leaninl|ZFSet| to the host's types, along equivalences only; the lemmas of \leaninl|transfer_simps| are its transfer rules, and the congruence step over arrows and products its decomposition.
The closest Lean antecedent is Lean~3's \leaninl|equiv_rw|~\cite{equiv-rw}, which transports goals and hypotheses along a given \leaninl|Equiv| and synthesizes congruences under type constructors by theorem search.
\leaninl|transfer| trades that generality for a fixed grammar and registered instances, so a failure is a visible residual goal; Lean~3's relator vocabulary remains the way to weaker relations.

\paragraph*{Mizar, Lisa, and Metamath}
Mizar~\cite{mizar-sota,mizar-overview} axiomatizes Tarski--Grothendieck set theory in first-order logic and develops its library entirely at the set level, with a ``soft'' type system of attributes and registrations layered on top, much as our thin typed interface does for Lean; Lisa~\cite{lisa-itp2023} builds axiomatic set theory on a schematic first-order kernel in Scala; and Metamath's \texttt{set.mm}~\cite{metamath} develops ZFC in a substitution-only kernel, with neither typed interface nor tactics.
All three keep mathematics at the set level, as we do; \zflean instead embeds the set-level layer in a dependently typed host, reusing its binders, notation, tactics, type classes, and algebraic hierarchy through explicit transfers.

\paragraph*{Typed encodings in Mathlib}
The objects of \Cref{sec:zflean-rel} have typed encodings in Mathlib: \leaninl|PFun α β|, \leaninl|Option|- or \leaninl|Part|-valued functions, \leaninl|Set α| and \leaninl|Finset α| over a fixed carrier, and extensional equality through \leaninl|funext|; the transfers of \Cref{subsec:bridges} reach them where they pay off.
They do not provide the setting of \Cref{sec:zflean-case-study}: a store that is a set of pairs, its domain a set, and itself a member of the set of all stores, so that well-definedness is a domain membership rather than a case analysis on \leaninl|Option|; for denotational semantics, that setting is itself the object of study.
A typed encoding also needs a carrier type before any denotation can be written, since a \leaninl|Set α| is bounded by its \leaninl|α|.
When the domain refers to itself in a way that is not strictly positive, as in B, Z, and TLA\textsuperscript{+} (whose values include the functions between values, informally $D \simeq V_0 \cup (D \rightharpoonup D)$, an equation no set satisfies), the naive carrier fails Lean's strict-positivity condition and another construction must stand in for it; a set-theoretic semantics does not solve the equation but dissolves it: values are sets of the universe, functions among them are sets of pairs, and no carrier collects them.

\paragraph*{Set theory in Lean}
\zflean builds on the ZFC model Carneiro contributed to Mathlib~\cite{mathlib-zfc}.
Han and van Doorn's formalization of forcing and the independence of the continuum hypothesis~\cite{han-vandoorn-forcing,han-vandoorn-independence} is the most substantial development on its encoding; it is geared toward one metamathematical goal, whereas \zflean aims at a reusable interface for everyday set-level mathematics.
Carneiro's translation of Metamath developments into Lean over the same model, with a set/type bridging layer~\cite{carneiro-mm-lean}, is machine-generated and not a user-facing library.
SynthLean~\cite{synthlean}, a certifying proof assistant for synthetic mathematics inside Lean, also hosts a second foundation in Lean; here the guest is classical ZFC and the emphasis is on everyday ergonomics rather than certification.

\paragraph*{ETCS and categorical foundations}
Lawvere's Elementary Theory of the Category of Sets~\cite{lawvere-etcs} axiomatizes the well-pointed category of sets with finite limits, cartesian closure, a natural numbers object, and choice, with extensional equality at the morphism level and ``up to isomorphism'' reasoning built in; standard ETCS, without replacement strengthenings, is weaker than ZFC, and more a metatheoretic object than a framework for set-level development.

\zflean shares with the above works the goal of enabling set-level mathematics inside a typed proof assistant; it differs in combining a rewriting-friendly relational calculus, structural side-condition automation, and explicit transfers to the host's typed libraries, with a tactic that moves goals across them.

\section{Conclusion}
\label{sec:zflean-conclusion}

We presented \zflean, a Lean~4 library for set-level mathematics on top of Mathlib's ZFC model: a relational calculus with bounded automation for its side conditions, canonical constructions up to $\mathbb{Q}$ with their algebraic structures and set-level anchors, and transfers to native types.
A denotational semantics written as a client shows the layer at work: side conditions discharged by the tactics, definedness as domain membership, program equivalence as equality of sets.
Future work goes to the reals, to automation for the weakening and destructuring steps where the closure search stops, and to transfer along embeddings rather than equivalences.


\bibliographystyle{ACM-Reference-Format}
\bibliography{ref}

\end{document}